\documentclass[12pt]{article}
\pdfoutput=1

\usepackage[margin=1in]{geometry}

\usepackage[T1]{fontenc}
\usepackage[utf8]{inputenc}
\usepackage[english]{babel}
\usepackage{authblk}
\usepackage{mathpazo}
\usepackage{amsfonts}
\usepackage{latexsym}
\usepackage{amssymb}
\usepackage{float} 

\usepackage{hyperref}
\hypersetup{
  pdfpagemode=UseNone,
  colorlinks=true,
  citecolor=blue,
  linkcolor=blue,
  urlcolor=cyan
}

\usepackage{amsmath}
\usepackage{mathtools}
\usepackage[roman]{complexity}
\newclass{\coQAM}{coQAM}
\newclass{\coQIP}{coQIP}
\newclass{\IPBQP}{\IP_\BQP}
\newclass{\UQCMA}{UQCMA}

\usepackage{amsthm}
\newtheorem{theorem}{Theorem}
\newtheorem{lemma}[theorem]{Lemma}

\theoremstyle{definition}

\usepackage{tikz,appendix}

\newcommand{\tinyspace}{\mspace{1mu}}
\newcommand{\microspace}{\mspace{0.5mu}}

\newcommand{\op}[1]{\operatorname{#1}}

\renewcommand{\int}{\operatorname{int}}

\newcommand{\fid}{\operatorname{F}}

\newcommand\I{\mathbb{1}}

\renewcommand{\natural}{\mathbb{N}}

\newcommand{\abs}[1]{\lvert #1 \rvert}
\newcommand{\bigabs}[1]{\bigl\lvert #1 \bigr\rvert}

\newcommand{\biggabs}[1]{\biggl\lvert #1 \biggr\rvert}

\newcommand{\ip}[2]{\langle #1 \vert #2\rangle}
\newcommand{\bigip}[2]{\bigl\langle #1 \big\vert #2 \bigr\rangle}

\newcommand{\norm}[1]{\lVert\tinyspace #1 \tinyspace\rVert}
\newcommand{\bignorm}[1]{\bigl\lVert\tinyspace #1 \tinyspace\bigr\rVert}
\newcommand{\Bignorm}[1]{\Bigl\lVert\tinyspace #1 \tinyspace\Bigr\rVert}

\newcommand{\ket}[1]{\lvert\microspace #1 \microspace \rangle}
\newcommand{\bigket}[1]{\bigl\lvert\microspace #1 \microspace \bigr\rangle}

\newcommand{\bra}[1]{\langle\microspace #1 \microspace \rvert}
\newcommand{\bigbra}[1]{\bigl\langle\microspace #1 \microspace \bigr\rvert}

\newcommand{\ketbra}[2]{\ket{#1}\!\bra{#2}}

\renewcommand{\A}{\mathcal{A}}
\newcommand{\B}{\mathcal{B}}
\newcommand{\Q}{\mathcal{Q}}
\renewcommand{\R}{\mathcal{R}}

\newcommand{\tnorm}[1]{\norm{#1}_{1}}
\newcommand{\bigtnorm}[1]{\bigl\lVert\tinyspace #1 \tinyspace\bigr\rVert_{1}}
\newcommand{\Bigtnorm}[1]{\Bigl\lVert\tinyspace #1 \tinyspace\Bigr\rVert_{1}}

\newcommand{\floor}[1]{\lfloor #1 \rfloor}

\newcommand{\prob}[1]{\textup{#1}}

\newcommand{\QSD}{\prob{QSD}}

\renewcommand{\footnotemark}{}

\newcommand{\etal}{{\emph{et al.}}}

\newcommand{\arxiv}[1]{arXiv:\href{https://arxiv.com/abs/#1}{#1}}

\newcommand{\myref}[2]{\hyperref[#2]{#1~\ref{#2}}}

\begin{document}

\title{\textbf{\Large Oracle Separations for\\ Quantum Statistical
    Zero-Knowledge}}

\author[1]{Sanketh Menda}
\author[1,2]{John Watrous}

\affil[1]{%
  Institute for Quantum Computing and School of Computer
  Science\protect\\
  University of Waterloo, Canada\vspace{2mm}}

\affil[2]{%
  Canadian Institute for Advanced Research\protect\\
  Toronto, Canada}
  
\date{January 26, 2018}

\renewcommand\Affilfont{\normalsize\itshape}
\renewcommand\Authfont{\large}
\setlength{\affilsep}{6mm}
\renewcommand\Authand{\rule{6mm}{0mm}and\rule{6mm}{0mm}}

\maketitle
\begin{abstract}
  This paper investigates the power of quantum statistical zero knowledge
  interactive proof systems in the relativized setting.
  We prove the existence of an oracle relative to which quantum statistical
  zero knowledge does not contain $\UP \cap \coUP$, and we prove that quantum
  statistical zero knowledge does not contain $\UP$ relative to a random
  oracle with probability~1.
  Our proofs of these statements rely on a bound on output state discrimination
  for relativized quantum circuits based on the quantum adversary
  method of Ambainis~\cite{Ambai2002}, following a technique similar to one
  used by Ben-David and Kothari \cite{Ben-daK2017} to prove limitations on
  a query complexity variant of quantum statistical zero-knowledge.
\end{abstract}

%-----------------------------------------------------------------------------%
\section{Introduction}
%-----------------------------------------------------------------------------%

Interactive proof systems, first introduced by Goldwasser, Micali, and
Rackoff~\cite{GoldwMR1985,GoldwMR1989} and
Babai~\cite{Babai1985,BabaiM1988}, form a cornerstone of complexity theory.
Many variants of interactive proof systems have been studied, including quantum
statistical zero-knowledge interactive proof systems
\cite{Watro2002,Kobay2003,Watro2009,HaydeMW2013,GutosHMW2015,Chen2016},
which are the topic of this paper.

An interactive proof system has the property of being
\emph{statistical zero-knowledge} if the prover does not ``leak'' statistically
significant knowledge to a computationally bounded verifier on positive problem
inputs.
It is known that the class $\QSZK$ of decision problems having quantum
statistical zero-knowledge interactive proof systems is closed under
complementation and is contained in $\QIP(2)$, the class of decision problems
having (not necessarily zero-knowledge) quantum interactive proof systems in
which precisely two messages are exchanged between the prover and
verifier.
Unlike its classical counterpart $\SZK$, however, it is not known if the
containment of $\NP$ in $\QSZK$ has unexpected complexity theoretic
consequences.
(The containment of $\NP$ in $\SZK$ implies that the polynomial-time hierarchy
collapses to $\AM$ \cite{Fortn1989,AiellH1987,BoppaHZ1987}.)

In this paper we consider $\QSZK$ in a relativized setting, with the aim of
proving limitations on the power of this class.
We prove two results along these lines.
First, we prove that there exists an oracle relative to which $\UP \cap \coUP$
is not contained in $\QSZK$, where $\UP$ is a restricted variant of $\NP$
containing decision problems recognized by a polynomial-time nondeterministic
Turing machine with no more than one accepting computation path on every
valid input.
Second, we prove that with respect to a random oracle, the class $\UP$ is not
contained in $\QSZK$ with probability~1.

Our proofs make use of the positive weights quantum adversary method of
Ambainis~\cite{Ambai2002}.
The positive weights quantum adversary method is known to not always give tight
bounds on quantum query complexity, see
\cite{AaronS2004, Zhang2005, SpaleS2006}, but it suffices for our needs. 
Ben-David and Kothari~\cite{Ben-daK2017} recently observed that the positive
weights quantum adversary method can be used to prove limitations on a query
complexity variant of quantum statistical zero-knowledge.
Also, a related notion of state conversion in query complexity was investigated
by Lee \etal{}~\cite{LeeMRSS2011}. 

%-----------------------------------------------------------------------------%
\section{Preliminaries}\label{sec:preliminaries}
%-----------------------------------------------------------------------------%

In this section, we summarize relevant concepts regarding complexity theory and
quantum computation, with which we assume the reader is generally familiar.

\subsubsection*{Measures of distance between quantum states}

We define the \emph{trace norm} $\norm{A}_1$ of an operator $A$ as the sum of
its singular values (with no pre-factor of 1/2), and we define the
\emph{fidelity} between quantum states $\rho$ and $\sigma$ as
\begin{equation}
  \fid(\rho,\sigma) = \bignorm{\sqrt{\rho}\sqrt{\sigma}}_1
\end{equation}
(with the right-hand side not being squared).
Uhlmann's theorem \cite{Uhlma1976} implies that the fidelity between two
states $\rho$ and $\sigma$ is given by
\begin{equation}
  \label{eq:Fidelity}
  \fid(\rho,\sigma) =
  \max_{\ket{\psi},\ket{\varphi}} \abs{\langle \psi\!\mid\!\phi \rangle},
\end{equation}
where the maximization is over all purifications
$\ket{\psi}$ and $\ket{\phi}$ of $\rho$ and $\sigma$, respectively.
For quantum states $\rho$ and $\sigma$, their trace distance and fidelity are
related by the Fuchs\--van de Graaf inequalities \cite{FuchsG1999}
as follows:
\begin{equation}
  \label{eq:FvdG} 2 - 2\fid(\rho,\sigma) \leq \tnorm{\rho-\sigma} \leq
  2\sqrt{1-\fid(\rho,\sigma)^2}.
\end{equation}

\subsubsection*{Quantum Circuits}

The results we prove in this paper are not sensitive to the specific gate set
one chooses to adopt when discussing quantum circuits.
Nevertheless, for the sake of simplicity and concreteness, we may assume that
quantum circuits in this paper are composed of \emph{Hadamard},
\emph{Toffoli}, and \emph{phase shift} gates, as well as \emph{query gates}
(discussed below).

We may also assume that the introduction of new, initialized qubits into a
quantum circuit are represented by \emph{auxiliary qubit} gates, which have no
inputs and output one qubit in the state $\ket{0}$;
and circuits may also include \emph{erasure} gates, which take one qubit as
input and have no outputs, effectively tracing out their input qubit.
Of course these gates can be removed from a circuit, provided that a suitable
number of qubits in the state $\ket{0}$ are provided as part of the input into
the circuit and that the qubits that would have gone into erasure gates are
traced-out once the computation is finished---and this is what we mean when we
refer to a \emph{unitary purification} of a given circuit.
It is, however, convenient to view auxiliary qubit gates and erasure gates as
being gates, so that we may speak of circuits that have no inputs and output
some number of qubits in a possibly mixed state.
We refer the reader to \cite{AharoKN1998} and \cite{Watro2011} for further
details on quantum circuits acting on mixed states.

\subsubsection*{Oracles and relativization}

Throughout the paper we denote the binary alphabet by $\Sigma = \{0,1\}$.
An \emph{oracle} is any subset $A\subseteq\Sigma^{\ast}$ of binary strings, to
which membership queries are made available at unit cost.
We will use the term \emph{black box} to refer to the restriction of an
oracle to strings of a single, fixed length.

With respect to a given black box $B\subseteq\Sigma^n$, the corresponding query
gate $K_B$ is the $(n+1)$-qubit unitary gate defined by the following action on
the standard basis:
\begin{equation}
  K_B\ket{x}\ket{a} =
  \begin{cases}
    \ket{x}\ket{\neg a} & \text{if}\;x\in B\\ \ket{x}\ket{a}
    & \text{if}\;x\not\in B,
  \end{cases}
\end{equation}
for all $x\in\{0,1\}^n$ and $a\in\{0,1\}$.
Equivalently, one may write
\begin{equation}
  \label{eq:oracle-formula}
  K_B = \sum_{x\in B} \ketbra{x}{x} \otimes
  X + \sum_{x\in \overline{B}} \ketbra{x}{x} \otimes \I,
\end{equation}
where (in this case) $\I$ denotes the identity operator on a single qubit and
$X$ denotes a single-qubit \texttt{NOT} operation.

A \emph{relativized circuit} is one that may include query gates
(accessing one black box for each string length), and one views that such a
circuit queries a given oracle if the query gates are consistent with the
oracle.

\subsubsection*{Honest-verifier quantum statistical zero-knowledge}

With respect to quantum statistical zero-knowledge, we will focus on
the \emph{honest-verifier} definition of this class, which is simpler
to state than the more cryptographically satisfying
\emph{general-verifier} definition.  Although the two definitions are
known to be equivalent \cite{Watro2009}, it is only the easier of the
two containments needed to prove this equivalence that is relevant to
our results.  That is, because we prove that certain relativized
languages are not contained in $\QSZK$, no generality is lost in
making use of the honest-verifier definition, as it gives rise to a
complexity class that is at least as large as the one given by the
general-verifier definition.

With respect to an oracle $A$, a language $L$ is in $\QSZK^A$ if there
exists a quantum interactive proof system $(V,P)$ satisfying the following
(somewhat informally stated) properties:
\begin{enumerate}
\item[1.]
  The verifier is \emph{efficient}:
  $V$ is specified by a polynomial-time generated family of tuples of quantum
  circuits that represent the verifier's actions.
  These circuits may make queries to the oracle $A$.
\item[2.]
  The proof system is \emph{complete and sound}:
  on inputs in $L$, the prover $P$ (which may also query the oracle $A$)
  causes $V$ to accept with high probability, and on inputs not in $L$,
  no prover causes $V$ to accept, except with small probability.
\item[3.]
  The proof system is \emph{honest-verifier quantum statistical
    zero-knowledge}: on inputs in $L$, and assuming that one considers a
  unitary purification of $V$, the \emph{view} of $V$ (represented
  by the tensor product of the states it holds after each message exchange
  takes place) has negligible trace distance to a state that can be produced by
  a polynomial-time uniform family of quantum circuits that do not interact
  with a prover (but that may make queries to $A$).
\end{enumerate}

For the purposes of this paper, it is not necessary for us to make use of the
specific details of the definition just suggested---we instead rely on the
existence of a complete promise problem for $\QSZK$, known as
\textsc{Quantum State Distinguishability} \cite{Watro2002}.
A relativized version of this problem can be phrased as follows.\vspace{2mm}

\begin{table}[H]
  \begin{center}
    \textsc{Relativized Quantum State Distinguishability} ($\QSD^A$)
  \end{center}
  \noindent
  \begin{tabular}[H]{l@{\hspace{5mm}}p{13.75cm}}
    Input: & Relativized quantum circuits $Q_0$ and $Q_1$ that take no input
    qubits and produce output states on the same number of qubits.
    Let $\rho_0(A)$ and $\rho_1(A)$ denote the states produced by $Q_0$ and
    $Q_1$, respectively, when the query gates of these circuits operate in
    accordance with the oracle $A$.\\[1mm]
    Yes: & $(Q_0,Q_1)$ is a \emph{yes-instance} of $\QSD^A$, denoted
    $(Q_0,Q_1) \in \QSD^A_{\textup{yes}}$, if $\rho_0(A)$ and $\rho_1(A)$ are
    \emph{far}:
    \[
    \frac{1}{2} \tnorm{\rho_0(A) - \rho_1(A)} \geq \frac{2}{3}.
    \]\\
    No: & $(Q_0,Q_1)$ is a \emph{no-instance} of $\QSD^A$, denoted
    $(Q_0,Q_1) \in \QSD^A_{\textup{no}}$, if $\rho_0(A)$ and $\rho_1(A)$ are
    \emph{close}:
    \[
    \frac{1}{2} \tnorm{\rho_0(A) - \rho_1(A)} \leq \frac{1}{3}.
    \]
  \end{tabular}
\end{table}
\vspace{-3mm}

Although the proof that $\QSD$ is complete for $\QSZK$ found in
\cite{Watro2002} does not mention query gates, the proof does extend
directly to the relativized setting; query gates can simply be treated in the
same way as other gates within the context of this proof.
The following theorem expresses this fact in a form that is convenient for the
purposes of this paper.
(In this theorem, $\Gamma$ denotes an arbitrary alphabet over which languages
are to be considered---but we will only need to concern ourselves with
the unary alphabet $\Gamma = \{0\}$ in this paper.)

\begin{theorem}\label{thm:QSD-QSZK-complete}
  Let $L\subseteq\Gamma^{\ast}$ be a language and let $A\subseteq\Sigma^{\ast}$
  be an oracle.
  The language $L$ is contained in $\QSZK^A$ if and only if there exists
  a polynomial-time uniform family of pairs of relativized quantum circuits
  $\{(Q^x_0,Q^x_1)\,:\,x\in\Gamma^{\ast}\}$ with these properties:
  \begin{enumerate}
  \item If $x \in L$, then $(Q^x_0,Q^x_1) \in \QSD^A_{\textup{yes}}$, and
  \item If $x \not\in L$, then $(Q^x_0,Q^x_1) \in \QSD^A_{\textup{no}}$.
  \end{enumerate}
\end{theorem}

\subsubsection*{Unambiguous polynomial-time}

Finally, the class $\UP$, which stands for \emph{unambiguous polynomial time},
is a restricted variant of $\NP$ that was first defined by
Valiant~\cite{Valia1976}.
A language $L$ is in the class $\UP$ if there exists a polynomial-time
nondeterministic Turing machine $M$ satisfying these conditions:
\begin{enumerate}
\item If $x \in L$, then $M$ has exactly one accepting computation path on
  input $x$.
\item If $x \not\in L$, then $M$ has no accepting computation paths on input
  $x$.
\end{enumerate}
Relativized variants of $\UP$ are defined in the natural way, by allowing the
machine $M$ to make oracle queries.

%-----------------------------------------------------------------------------%
\section{Adversary bound for output state discrimination}\label{sec:adversary}
%-----------------------------------------------------------------------------%

In this section we prove a lemma that will be used to prove that certain
problems fall outside of $\QSZK$ relative to some oracles.

Before proving the main lemma, we will prove a somewhat more basic lemma that
implies that a quantum circuit must, on average, make a large number of
queries to a black box in order to produce output states that allow one to
discriminate between an empty black box and a black box containing one string.
The proof makes use of the positive weights quantum adversary
method~\cite{Ambai2002}.

\begin{lemma}
  \label{FidelityBound}
  Let $Q$ be a quantum circuit that takes no input and makes $T$ queries to a
  $n$\--bit black box and let $\rho(B)$ denote the output of $Q$ when the black
  box is described by $B \subseteq \Sigma^n$.
  It holds that
  \begin{equation}
    \frac{1}{2^n} \sum_{x \in \Sigma^n}
    \fid\bigl(\rho(\{x\}), \rho(\varnothing)\bigr) \geq 1 - \frac{2T}{2^{n/2}}.
  \end{equation}
\end{lemma}

\begin{proof}
  Let $R$ be a unitary quantum circuit that purifies $Q$.
  For a given black box $B$, the unitary operator corresponding to the action
  of $R$ can be expressed as
  \begin{equation}
    U_T(K_B \otimes \I)U_{T-1}(K_B \otimes \I)\cdots
    U_1(K_B \otimes \I)U_0
  \end{equation}
  where each $U_t$ is a unitary operator that is independent of $B$ (and $K_B$
  represents a query gate to $B$ as already mentioned).
  Let $\ket{\psi_t(B)}$ represent the state immediately after the unitary
  operation $U_t$ is performed, assuming the computation begins with all qubits
  initialized to the $\ket{0}$ state:
  \begin{equation}
    \ket{\psi_t(B)} = U_t(K_B \otimes \I)U_{t-1}(K_B \otimes\I)\cdots
    U_0\ket{0 \cdots 0}.
  \end{equation}

  Next, define a progress function
  \begin{equation}
    f(t) = \sum_{x \in \Sigma^n}
    \bigabs{\bigip{\psi_t(\{x\})}{\psi_t(\varnothing)}}
  \end{equation}
  for all $t\in\{0,\ldots,T\}$.
  Because $R$ purifies $Q$, it holds that $\ket{\psi_T(B)}$ purifies
  $\rho(B)$ (for any choice of a black box $B$), and therefore
  \begin{equation}\label{eq:Rel-f-fid}
    f(T) \leq \sum_{x\in\Sigma^n}
    \fid\bigl(\rho(\{x\}),\rho(\varnothing)\bigr)
  \end{equation}
  by the fact that the fidelity function is non-decreasing under partial
  tracing.
  It holds that $\ket{\psi_0(\{x\})} = \ket{\psi_0(\varnothing)}$, and
  therefore
  \begin{equation}\label{eq:f-zero}
    f(0) = \sum_{x \in \Sigma^n}
    \abs{\ip{\psi_0(\{x\})}{\psi_0(\varnothing)}} = 2^n.
  \end{equation}
  As
  \begin{equation}
    \ket{\psi_{t+1}(B)} = U_{t+1}(K_B \otimes \I)\ket{\psi_t(B)},
  \end{equation}
  it follows that
  \begin{equation}
    \bigabs{\bigip{\psi_{t+1}(\{x\})}{\psi_{t+1}(\varnothing)}}
    = \bigabs{\bigbra{\psi_t(\{x\})} K_{\{x\}}\otimes\I
      \bigket{\psi_t(\varnothing)}}.
  \end{equation}
  Making use of the expression \eqref{eq:oracle-formula}, one finds that
  \begin{equation}
    \begin{multlined}
      \bigbra{\psi_t(\{x\})} K_{\{x\}} \otimes \I
      \bigket{\psi_t(\varnothing)}\\[1mm]
      = \bigbra{\psi_t(\{x\})} \ketbra{x}{x}
      \otimes (X-\I) \otimes \I\, \bigket{\psi_t(\varnothing)}
      + \bigip{\psi_t(\{x\})}{\psi_t(\varnothing)}.
    \end{multlined}
  \end{equation}
  Therefore, by the Cauchy--Schwarz and triangle inequalities, and
  making use of the fact that $\norm{X - \I} = 2$, one obtains
  \begin{equation}
    \begin{multlined}
      \bigabs{\bigbra{\psi_t(\{x\})} K_{\{x\}} \otimes
        \I \bigket{\psi_t(\varnothing)}}\\[1mm]
      \geq \bigabs{\bigip{\psi_t(\{x\})}{\psi_t(\varnothing)}}
      - 2\bignorm{(\bra{x} \otimes \I \otimes \I)\ket{\psi_t(\varnothing)}}.
    \end{multlined}
  \end{equation}
  Using the Cauchy--Schwarz inequality again, it follows that
  \begin{equation}
    \begin{aligned}
      f(t+1) & = \sum_{x\in\Sigma^n}
      \bigabs{\bigip{\psi_{t+1}(\{x\})}{\psi_{t+1}(\varnothing)}}\\
      & \geq f(t) - 2\sum_{x\in\Sigma^n}
      \bignorm{(\bra{x} \otimes \I \otimes \I)\bigket{\psi_t(\varnothing)}}\\
      & \geq f(t) - 2\cdot 2^{n/2}.
    \end{aligned}
  \end{equation}
  Consequently,
  \begin{equation} f(T) = f(0) + \sum_{t = 0}^{T-1} (f(t+1) - f(t)) \geq
    2^n - 2T\cdot 2^{n/2}.
  \end{equation}
  Finally, using relation \eqref{eq:Rel-f-fid} one finds that
  \begin{equation} \frac{1}{2^n}\sum_{x \in \Sigma^n}
    \fid\bigl(\rho(\{x\}),\rho(\varnothing)\bigr) \geq 1 - \frac{2T}{2^{n/2}}
  \end{equation}
  as required.
\end{proof}

We now present the main lemma, which is proved through the use of
\myref{Lemma}{FidelityBound} along with standard arguments.

\begin{lemma}[Main Lemma]
  \label{lemma:main}
  Let $Q_0$ and $Q_1$ be quantum circuits, both taking no input qubits,
  producing the same number of output qubits, and making at most $T$ queries to
  an $n$\--bit black box, and let $\rho_0(B)$ and $\rho_1(B)$ denote the output
  states of these circuits when the black box is described by
  $B\subseteq\Sigma^n$.
  If $T$ and $n$ satisfy
  \begin{equation}
    T \leq \frac{\sqrt{2^n}}{20736},
  \end{equation}
  then there are at least $\frac{2}{3}2^n$ distinct choices of a
  string $x\in\Sigma^n$ such that
  \begin{equation}
    \biggabs{\frac{1}{2} \Bigtnorm{\rho_0(\{x\}) - \rho_1(\{x\})} - 
      \frac{1}{2} \Bigtnorm{\rho_0(\varnothing) - \rho_1(\varnothing)}}
      < \frac{1}{6}.
  \end{equation}
\end{lemma}

\begin{proof}
  Define sets $S_0,S_1\subseteq\Sigma^n$ as follows:
  \begin{equation}
    \begin{aligned}
      S_0 & = \biggl\{ x \in \Sigma^n \;:\;
      \bignorm{\rho_0(\{x\}) - \rho_0(\varnothing)}_1<\frac{1}{6}\biggr\},\\
      S_1 & = \biggl\{ x \in \Sigma^n \;:\;
      \bignorm{\rho_1(\{x\}) - \rho_1(\varnothing)}_1 < \frac{1}{6}
      \biggr\}.
    \end{aligned}
  \end{equation}
  Also define $S = S_0\cap S_1$.
  For every $x\in S$, it follows from the triangle inequality that
  \begin{equation}
    \begin{multlined}
      \biggabs{ \frac{1}{2}\Bignorm{\rho_0(\{x\}) - \rho_1(\{x\})}_1
        - \frac{1}{2}\Bignorm{\rho_0(\varnothing) - \rho_1(\varnothing)}_1}\\
      \leq \frac{1}{2}\Bignorm{\rho_0(\{x\}) - \rho_0(\varnothing)}_1
      + \frac{1}{2}\Bignorm{\rho_1(\{x\}) - \rho_1(\varnothing)}_1
      < \frac{1}{6}.
    \end{multlined}
  \end{equation}

  We will now prove that
  \begin{equation}
    \abs{S_0} \geq \frac{5}{6} 2^n
    \quad\text{and}\quad
    \abs{S_1} \geq \frac{5}{6}2^n.
  \end{equation}
  The same argument, applied separately to $Q_0$ and $Q_1$, establishes both
  inequalities, so let us focus on $Q_0$ and prove the first inequality.
  By making use of the Fuchs--van de Graaf inequalities, we conclude from
  \myref{Lemma}{FidelityBound} that
  \begin{equation}
    \frac{1}{2^n}\sum_{x \in \Sigma^n} \tnorm{\rho_0(\{x\}) -
      \rho_0(\varnothing)} \leq 4\sqrt{\frac{T}{2^{n/2}}}.
  \end{equation}
  Considering only those strings not contained in $S_0$ yields
  \begin{equation}
    \frac{2^n - \abs{S_0}}{6\cdot 2^n}
    \leq
    \frac{1}{2^n}
    \sum_{x\not\in S_0}\tnorm{\rho_0(\{x\}) - \rho_0(\varnothing)}
    \leq 4 \sqrt{\frac{T}{2^{n/2}}},
  \end{equation}
  which yields the required bound given the assumptions of the lemma.

  By the union bound there are at most $2^n/3$ strings that are either not
  contained in $S_0$ or not contained in $S_1$, which implies that
  $\abs{S} \geq \frac{2}{3}2^n$, as required.
\end{proof}

%-----------------------------------------------------------------------------%
\section{Oracle Separations}\label{sec:oracle-separations}
%-----------------------------------------------------------------------------%

In this section we apply the main lemma proved in the previous section to prove
the existence of oracles that establish limitations on the class $\QSZK$.

%-----------------------------------------------------------------------------%
\subsection{Separating UP intersect coUP from
  QSZK}\label{sec:UP-cap-coUP-separation}
%-----------------------------------------------------------------------------%

We begin by proving the existence of an oracle relative to which
$\UP \cap \coUP$ is not contained in $\QSZK$.

\begin{theorem}\label{thm:QSZK-UP-oracle}
  There exists an oracle $A$ for which
  $(\UP^A \cap \coUP^A)\not\subseteq \QSZK^A$.
\end{theorem}

\noindent
The remainder of this subsection is devoted to a proof of this theorem,
divided according to the main steps of the proof.

%-----------------------------------------------------------------------------%
\subsubsection*{Problem specification and inclusion in UP intersect coUP}
%-----------------------------------------------------------------------------%

The basic idea of the proof is to consider oracles that contain exactly one
string of each length along with the computational problem of determining the
first bit of this unique string for a given length.
Fortnow and Rogers \cite{FortnR1999} proved that a $\BQP$ machine requires
an exponential number of queries to solve this problem, and our proof
represents an extension of this argument to $\QSZK$.
In essence, our proof replaces their use of Lemma~4.7 of Bennett et
al.~\cite{BenneBBV1997} with \myref{Lemma}{lemma:main}.

The set of oracles under consideration is
\begin{equation}
  \A = \{ A \subseteq \Sigma^* :
  \text{$\abs{A \cap \Sigma^n} = 1$ for all $n\geq 1$} \},
\end{equation}
and, for any oracle $A \in \A$, we define a language $L(A)$ over the
single-letter alphabet $\Gamma = \{0\}$ as
\begin{equation}
  L(A) = \bigl\{0^n : \text{
    $n\geq 1$ and $1y \in A$ for some $y \in \Sigma^{n-1}$} \bigr\}.
\end{equation}
Our aim is to prove that there exists an oracle $A\in\A$ such that
\begin{equation}
  L(A)\in \UP^A \cap \coUP^A
\end{equation}
but
\begin{equation}
  L(A) \not\in \QSZK^A.
\end{equation}

Observe that for any oracle $A \in \A$, the inclusion $L(A) \in \UP^A$ is
established by the simple nondeterministic procedure presented in
\myref{Figure}{fig:L-in-UP}.
\begin{figure}[t]
  \noindent\hrulefill\vspace{-2mm}
  
  \begin{tabbing}
    x\=xxx\=xxx\=\kill
    \> On input $w = 0^n$:\\[2mm]
    \>\> \emph{Reject} if $n=0$.\\[1mm]
    \>\> Nondeterministically choose a string $y \in \Sigma^{n-1}$.\\[1mm]
    \>\> If $1y\in A$ then \emph{accept} else \emph{reject}.
  \end{tabbing}
  \vspace{-3mm}
  \noindent\hrulefill

  \caption{Nondeterministic decision procedure for $L(A)$.}
  \label{fig:L-in-UP}
\end{figure}
As there is exactly one string of each positive length in an oracle
$A \in \A$, it holds that the compliment of the language $L(A)$ is given by
\begin{equation}
  \overline{L(A)} = \bigl\{0^n : \text{$n=0$ or $0y \in A$ for
    some $y \in \Sigma^{n-1}$}\bigr\}.
\end{equation}
A similar argument to the one just presented reveals that
$\overline{L(A)} \in \UP^A$ for any $A \in \A$, and therefore
$L(A) \in (\UP^A \cap \coUP^A)$ for all oracles $A\in\A$.

%-----------------------------------------------------------------------------%
\subsubsection*{Black Box Separation}
%-----------------------------------------------------------------------------%

It remains to prove that there exists an oracle $A\in\A$ for which
$L(A) \not\in \QSZK^A$.
This is done in two steps, the first of which is a black box separation
based on the main lemma proved in the previous section.

Fix a positive integer $n$, and let $Q_0$ and $Q_1$ be quantum circuits that
take no input and make at most $T$ queries to an $n$\--bit black box, and let
$\rho_0(B)$ and $\rho_1(B)$ denote the output states of these circuits when the
black box is described by $B\subseteq\Sigma^n$.
We will say that the pair $(Q_0,Q_1)$ is \emph{incorrect} for a given choice of
a black box $B = \{x\} \subset \Sigma^n$ if either of these conditions hold:
\begin{enumerate}
\item[1.]
  $x \in 1\,\Sigma^{n-1}$ and
  \begin{equation}
    \frac{1}{2}\bignorm{\rho_0(\{x\}) - \rho_1(\{x\})}_1 < \frac{2}{3}.
  \end{equation}
\item[2.]
  $x \in 0\,\Sigma^{n-1}$ and
  \begin{equation}
    \frac{1}{2}\bignorm{\rho_0(\{x\}) - \rho_1(\{x\})}_1 > \frac{1}{3}.
  \end{equation}
\end{enumerate}
Otherwise the pair $(Q_0,Q_1)$ is \emph{correct} for $B$.

Now suppose that $T$ and $n$ satisfy
\begin{equation}
  T \leq \frac{\sqrt{2^n}}{20736},
\end{equation}
and define
\begin{equation}
  S = \biggl\{ x \in \Sigma^n \;:\;
  \biggabs{ \frac{1}{2}\bignorm{\rho_0(\{x\}) - \rho_1(\{x\})}_1
    - \frac{1}{2}\bignorm{\rho_0(\varnothing) - \rho_1(\varnothing)}_1}
  < \frac{1}{6} \biggr\}.
\end{equation}
By \myref{Lemma}{lemma:main}, the set $S$ has cardinality at least
$\frac{2}{3}2^n$.
If it is the case that
\begin{equation}
  \label{eq:implicant1}
  \frac{1}{2}\bignorm{\rho_0(\varnothing) - \rho_1(\varnothing)}_1
  \leq \frac{1}{2},
\end{equation}
then there must therefore exist at least 
$\frac{2}{3}2^n - \frac{1}{2}2^n = \frac{1}{6}2^n$ choices of $x\in\Sigma^n$
such that the first condition listed above holds.
Similarly, if it is the case that
\begin{equation}
  \label{eq:implicant2}
  \frac{1}{2}\bignorm{\rho_0(\varnothing) - \rho_1(\varnothing)}_1
  \geq \frac{1}{2},
\end{equation}
then there must therefore exist at least 
$\frac{2}{3}2^n - \frac{1}{2}2^n = \frac{1}{6}2^n$ choices of $x\in\Sigma^n$
such that the second condition listed above holds.
One of the two implicants \eqref{eq:implicant1} and \eqref{eq:implicant2} must
hold, establishing that $(Q_0,Q_1)$ is incorrect for a uniformly chosen
black box $B = \{x\} \subset\Sigma^n$ with probability at least 1/6.

%-----------------------------------------------------------------------------%
\subsubsection*{Oracle Existence}
%-----------------------------------------------------------------------------%

To prove the existence of an oracle $A\in\A$ for which $L(A)\not\in\QSZK^A$, we
use the probabilistic method, along the lines of the random oracle methodology
of Bennett and Gill \cite{BenneG1981}.
Suppose that
\begin{equation}
  \Q = \bigl\{ \bigl( Q_0^n,Q_1^n\bigr) \,:\,n\in\natural\bigr\}
\end{equation}
is a polynomial-time uniform family of pairs of relativized quantum circuits,
and consider the performance of these circuits on an oracle $A\in\A$ chosen
uniformly---meaning that for each positive integer $n$, one string of length
$n$ is selected uniformly and included in $A$, with the random selections being
independent for different choices of $n$.

Let $\rho_0^n(A)$ and $\rho_1^n(A)$ and denote the states output by $Q_0^n$ and
$Q_1^n$, respectively, when the query gates in these circuits operate in a way
that is consistent with the oracle $A$.
For a given choice of $A\in\A$, the pair $(Q^n_0,Q^n_1)$ therefore incorrectly
determines membership of $1^n$ in $L(A)$, with respect to the characterization
of $\QSZK^A$ given by \myref{Theorem}{thm:QSD-QSZK-complete}, if either of these
conditions hold:
\begin{enumerate}
\item[1.]
  $A \cap \Sigma^n = \{1y\}$ for some $y\in\Sigma^{n-1}$ and
  \begin{equation}
    \frac{1}{2}\bignorm{\rho_0^n(A) - \rho_1^n(A)}_1 < \frac{2}{3}.
  \end{equation}
\item[2.]
  $A \cap \Sigma^n = \{0y\}$ for some $y\in\Sigma^{n-1}$ and
  \begin{equation}
    \frac{1}{2}\bignorm{\rho_0^n(A) - \rho_1^n(A)}_1 > \frac{1}{3}.
  \end{equation}
\end{enumerate}
We will also say that $\Q$ is \emph{incorrect} for $A\in\A$ if
$(Q_0^n,Q_1^n)$ is incorrect for $A$ for at least one choice of a positive
integer $n$.
Our aim is to prove that $\Q$ is incorrect with probability~1.

A small inconvenience arises at this point, which is that the circuits
$Q_0^n$ and $Q_1^n$ are permitted to include query gates for lengths
\emph{different} from $n$, and therefore the events that
$(Q_0^n,Q_1^n)$ is incorrect for different choices of $n$ are not necessarily
independent.
(Of course it is evident from the definition of $L(A)$ that this possibility is
not helpful for solving the problem at hand, but the point must be addressed
nevertheless.)
This inconvenience can be circumvented by making use of a general result of
Bennett and Gill, but in the present case a simple way to proceed is to define
a new family
\begin{equation}
  \R = \bigl\{ \bigl( R_0^n,R_1^n\bigr) \,:\,n\in\natural\bigr\}
\end{equation}
of quantum circuits that is identical to $\Q$ except that, for each $n$, each
of the query gates of $Q_0^n$ and $Q_1^n$ having size different from $n$ are
hard-coded.
The hard-codings are chosen so that the probability that $(R_0^n,R_1^n)$ is
incorrect for a random choice of a black box $B = \{x\}\subset\Sigma^n$ is
minimized.
It is evident that the probability that $\Q$ is incorrect is no smaller than
the probability that $\R$ is incorrect, for a random choice of $A\in\A$, and
we have independence among the events that $(R_0^n,R_1^n)$ is incorrect for a
random choice of $A\in\A$ over all choices of $n$.

By the assumption that $\Q$ is polynomial-time uniform, the circuits $R_0^n$
and $R_1^n$ include a number of $n$-bit query gates that is polynomial in $n$.
For all but finitely many choices of $n$, it must therefore hold that the
number $T$ of $n$-bit queries made by either $R_0^n$ or $R_1^n$ must satisfy
the bound $T \leq \sqrt{2^n}/ 20736$.
This implies that for all but finitely many choices of~$n$, the pair
$\bigl( R_0^n,R_1^n\bigr)$ is incorrect with probability at least 1/6.
By the independence of these events for different choices of $n$, it follows
that $\R$, and therefore $\Q$, is incorrect with probability 1.

Finally, because there are countably many polynomial-time uniform families of
pairs of relativized quantum circuits, there exists an oracle $A\in\A$ for
which $L(A)\not\in\QSZK^A$, as this is true for a random $A \in \A$ with
probability~1.

%-----------------------------------------------------------------------------%
\subsection{Random Oracle Separation}
%-----------------------------------------------------------------------------%

Next we consider the relationship between $\UP$ and $\QSZK$ relative to
a \emph{random oracle}, meaning that each individual string is included in the
oracle with probability 1/2, independent of every other string.
Specifically, we prove that relative to a random oracle, $\UP$ is not contained
in $\QSZK$ with probability~1.
Our proof follows the methodology introduced by Beigel \cite{Beige1989},
who proved various random oracle separations involving $\UP$ and its variants.

\begin{theorem}\label{thm:QSZK-UP-random}
  For a random oracle $A$, it holds that $\UP^A \not\subseteq \QSZK^A$
  with probability~1.
\end{theorem}

\noindent
Again, the remainder of the subsection is devoted to a proof of this theorem,
divided according to the main steps of the proof.

%-----------------------------------------------------------------------------%
\subsubsection*{Problem specification and inclusion in UP (with probability 1)}
%-----------------------------------------------------------------------------%

For a given positive integer $n$, we define $m = \floor{\log(n)}$ and
$N=2^m$, and for the remainder of the proof we will always consider $m$ and
$N$ to be defined in this way, as functions of a given positive integer $n$.

For every positive integer $n$ and for every $k\in\{0,\ldots,2^{N-m}\}$,
define $\B^n_k$ to be the set of all black boxes $B\subseteq\Sigma^n$ for which
there exist precisely $k$ distinct choices of $x\in\Sigma^{N-m}$ such that
$xy0^{n-N}\in B$ for all $y\in\Sigma^m$.
More succinctly,
\begin{equation}
  \B^n_k = \Bigl\{
  B\subseteq\Sigma^n\,:\,
  \bigabs{ \bigl\{
    x\in\Sigma^{N-m}\,:\,x\Sigma^m 0^{n-N}
    \subseteq B\bigr\}}=k\Bigr\}.
\end{equation}
These sets define a partition
\begin{equation}
  \label{eq:partition}
  \B^n_0 \cup \B^n_1 \cup \cdots \cup \B^n_{2^{N-m}}
\end{equation}
of the set of all $n$-bit black boxes.

Next, define a language $L(A)$, for every oracle $A\subseteq\Sigma^{\ast}$, as
follows:
\begin{equation}
  L(A) = \bigl\{ 0^n\,:\, n\geq 1 \;\text{and}\;A\cap\Sigma^n \in \B^n_1\bigr\}.
\end{equation}
It will be proved, with respect to a random choice of $A$, that
$L(A)\in\UP^A$ with probability~1 and $L(A)\in\QSZK^A$ with probability~0.
The first step of the proof, which establishes that $L(A)\in\UP^A$ with
probability~1, is a special case of the results of Beigel \cite{Beige1989}.
We include a proof, both for completeness and because the concepts and notation
required for the proof are useful in the second step of the proof that concerns
$\QSZK$.
  
First, fix a positive integer $n$, consider a random choice of 
$B\subseteq\Sigma^n$ (where each string of length $n$ is independently included
in $B$ with probability 1/2), and define indicator random variables
\begin{equation}
  Z_x = \begin{cases}
    1 & \text{if $x\Sigma^m 0^{n-N}\subseteq B$}\\
    0 & \text{otherwise}.
  \end{cases}
\end{equation}
for every string $x\in\Sigma^{N-m}$.
Also define
\begin{equation}
  Z = \sum_{x\in\Sigma^{N-m}} Z_x,
\end{equation}
and observe that the value taken by the random variable $Z$ corresponds to the
index of the set in the partition \eqref{eq:partition} to which $B$ belongs.
It is the case that
\begin{equation}
  \op{E}[Z_x] = 2^{-2^m} = 2^{-N}
\end{equation}
for every $x\in\Sigma^{N-m}$, and moreover
\begin{equation}
  \op{Pr}(Z = 0) = \bigl( 1 - 2^{-N}\bigr)^{2^{N-m}} > 1 - \frac{1}{N},
\end{equation}
where the inequality follows from the fact that $2^N > 2^m = N$ together
with the observation that the function $k\mapsto (1 - 1/k)^k$ is strictly
increasing.
We also have
\begin{equation}
  \op{Pr}(Z=1) = 2^{N-m}\cdot 2^{-N}\cdot\bigl(1-2^{-N}\bigr)^{2^{N-m}-1}
  > \frac{1}{N} - \frac{1}{N^2},
\end{equation}
and therefore
\begin{equation}
  \op{Pr}(Z \geq 2) < \frac{1}{N^2} \leq \frac{4}{n^2}.
\end{equation}

Now, the series
\begin{equation}
  \sum_{n=1}^{\infty} \frac{4}{n^2}
\end{equation}
converges, so it follows from the Borel--Cantelli lemma that for a random
oracle $A\subseteq\Sigma^{\ast}$, with probability 1 there are at most finitely
many values of $n$ for which
\begin{equation}
  \label{eq:B-bad-on-n}
  A\cap \Sigma^n \not\in \B^n_0 \cup \B^n_1.
\end{equation}
It therefore holds with probability 1 that $L(A) \in \UP^A$, for if there are
finitely many values of $n$ for which \eqref{eq:B-bad-on-n} holds, then
membership in $L(A)$ can be decided through the nondeterministic decision
procedure described in \myref{Figure}{fig:random-L-in-UP}.

\begin{figure}[t]
  \noindent\hrulefill
  
  \begin{tabbing}
    x\=xxx\=xxx\=\kill
    \> On input $w = 0^n$:\\[2mm]
    \>\> If $n=0$ or $n$ is one of the finitely many values for which
    $A\cap \Sigma^n \not\in \B^n_0 \cup \B^n_1$,\\
    \>\> then \emph{reject}.\\[2mm]
    \>\> Let $m = \lfloor\log(n)\rfloor$ and $N = 2^m$.\\[2mm]
    \>\> Nondeterministically choose a string $x \in \Sigma^{N-m}$.\\[2mm]
    \>\> If $xy0^{n-N} \in A$ for all $y\in\Sigma^m$ then \emph{accept},
    else \emph{reject}.
  \end{tabbing}
  \vspace{-2mm}
  \noindent\hrulefill
  \caption{A polynomial-time nondeterministic decision procedure for $L(A)$
    with either 0 or 1 accepting computation, provided that there are finitely
    many values of $n$ for which
    $A\cap \Sigma^n \not\in \B^n_0 \cup \B^n_1$.
  }
  \label{fig:random-L-in-UP}
\end{figure}

%-----------------------------------------------------------------------------%
\subsubsection*{Black box separation}
%-----------------------------------------------------------------------------%

It remains to prove that $L(A)\in\QSZK^A$ with probability~0.
The first step toward proving this fact is to consider a simple way of
modifying queries made by quantum circuits.

Fix a positive integer $n$, along with an arbitrary subset
$C\subseteq\Sigma^n$, let $m = \lfloor\log(n)\rfloor$ and $N = 2^m$ as
before, and suppose that $B\subseteq\Sigma^{N-m}$ is a given black box.
Using a single query to $B$, it is possible to design a circuit (into which $C$
may be hard-coded) that exactly simulates a query to the set
\begin{equation}
  \label{eq:black-box-from-A}
  C \cup B\Sigma^m 0^{n-N}.
\end{equation}

Now assume that $Q$ is a quantum circuit that takes no inputs and makes at most
$T$ queries to an $n$-bit black box, and as above suppose that
$C\subseteq\Sigma^n$ is a fixed subset of strings of length $n$ and
$B\subseteq\Sigma^{N-m}$ is an $(N-m)$-bit black box.
By replacing each query gate of $Q$ with the circuit suggested above, one
obtains a new circuit $R$ that makes $T$ queries to $B$ and produces exactly
the same output as $Q$ when run on the black box \eqref{eq:black-box-from-A}.

Next, suppose that $Q_0$ and $Q_1$ are two quantum circuits that take no
input and make at most $T$ queries to an $n$-bit black box, and let
$\rho_0(D)$ and $\rho_1(D)$ denote the outputs of these circuits on a given
black box $D\subseteq\Sigma^n$.
If it is the case that $T \leq \sqrt{2^{N-m}}/20736$, then for an arbitrary
choice of $C\subseteq\Sigma^n$, there are at least $\frac{2}{3}2^{N-m}$
distinct choices of a string $x\in\Sigma^{N-m}$ such that
\begin{equation}
  \label{eq:bound-on-R-circuits}
  \begin{multlined}
    \hspace{-1cm}
    \biggl|
    \frac{1}{2}\bigtnorm{
      \rho_0(C \cup x \Sigma^m 0^{n-N}) -
      \rho_1(C \cup x \Sigma^m 0^{n-N})}\\
    - \frac{1}{2} \bigtnorm{\rho_0(C) - \rho_1(C)}\biggr|
    < \frac{1}{6}.
    \hspace{-1cm}
  \end{multlined}
\end{equation}
This follows from \myref{Lemma}{lemma:main}, together with the observation
described in the previous paragraph.
That is, for the circuits $R_0$ and $R_1$ resulting from $Q_0$ and $Q_1$
together with the given choice of $C$ by the process above, we obtain
output states $\sigma_0(B)$ and $\sigma_1(B)$ (on a given black box
$B\subseteq\Sigma^{N-m}$) satisfying
\begin{equation}
  \begin{aligned}
    \sigma_0(\varnothing) & = \rho_0(C),\\
    \sigma_1(\varnothing) & = \rho_1(C),\\
    \sigma_0(\{x\}) & = \rho_0(C \cup x \Sigma^m 0^{n-N}),\\
    \sigma_1(\{x\}) & = \rho_1(C \cup x \Sigma^m 0^{n-N}).
  \end{aligned}
\end{equation}

Moving closer to the language $L(A)$, we may say that a pair of quantum
circuits $(Q_0,Q_1)$ that makes $n$-bit queries to a black box
$B\subseteq\Sigma^n$ is \emph{incorrect} for $B$ if one of the following two
statements is satisfied:
\begin{enumerate}
\item[1.]
  $B\in\B_0$ and $\norm{\rho_0(B) - \rho_1(B)}_1 > 1/3$.
\item[2.]
  $B\in\B_1$ and $\norm{\rho_0(B) - \rho_1(B)}_1 < 2/3$.
\end{enumerate}
Otherwise, $(Q_0,Q_1)$ is \emph{correct} for $B$.

Now consider a random choice of a black box $B\subseteq\Sigma^n$, with each
string being included in $B$ independently with probability 1/2.
Suppose $(Q_0,Q_1)$ is correct for a $\delta$ fraction of black boxes
in $\B^n_0$.
For each $C\in\B^n_0$ for which $(Q_0,Q_1)$ is correct, there are
at least $\frac{2}{3}2^{N-m}$ choices of $x\in\Sigma^{N-m}$ such that
$(Q_0,Q_1)$ is incorrect for $C \cup x \Sigma^m 0^{n-N}$ by the analysis
above.
Each element of $B \in \B^n_1$ can be obtained as
$B = C \cup x \Sigma^m 0^{n-N}$ for $2^N - 1$ distinct sets
$C\in\B^n_0$, and therefore $(Q_0,Q_1)$ is incorrect for at least
\begin{equation}
  (1 - \delta) \abs{\B^n_0} +
  \frac{2\cdot\delta\cdot 2^{N-m}}{3 \cdot 2^N}\abs{\B^n_0}
  \geq \frac{2 \abs{\B^n_0}}{3 N}
\end{equation}
distinct choices of $B \in \B^n_0 \cup \B^n_1$.
For a random choice of $B\subseteq\Sigma^n$, it therefore holds that
$(Q_0,Q_1)$ is incorrect with probability at least
\begin{equation}
  \frac{2}{3N} - \frac{2}{3N^2} \geq
  \frac{2}{3n} - \frac{2}{3n^2} \geq \frac{1}{3n},
\end{equation}
provided $n\geq 2$.

%-----------------------------------------------------------------------------%
\subsubsection*{Oracle Existence}
%-----------------------------------------------------------------------------%

Suppose that
\begin{equation}
  \Q = \bigl\{ \bigl( Q_0^n,Q_1^n\bigr) \,:\,n\in\natural\bigr\}
\end{equation}
is a polynomial-time uniform family of pairs of relativized quantum circuits,
and consider the performance of these circuits on a random oracle
$A\subseteq\Sigma^{\ast}$.
That is, we will consider the probability that each pair $(Q^n_0,Q^n_1)$
correctly determines membership in $L(A)$, with respect to the characterization
of $\QSZK$ given by \myref{Theorem}{thm:QSD-QSZK-complete}.

A similar issue to the one discussed in the previous subsection now arises,
due to the possibility for the circuits $Q_0^n$ and $Q_1^n$ to make queries
to $A$ on strings of length different from $n$, and the same argument allows
for this issue to be circumvented.
That is, there must exist a family
\begin{equation}
  \R = \bigl\{ \bigl( R_0^n,R_1^n\bigr) \,:\,n\in\natural\bigr\}
\end{equation}
of quantum circuits, where $R_0^n$ and $R_1^n$ include a number of $n$-bit
query gates that is polynomial in $n$ and include no query gates for strings of
other lengths, such that the probability $\Q$ is incorrect is at least the
probability $\R$ is incorrect, for a random oracle $A$.

For all but finitely many choices of $n$, it must therefore hold that the
number $T$ of $n$-bit queries made by either $R_0^n$ or $R_1^n$ must satisfy
the bound $T \leq \sqrt{2^{N-m}}/20736$.
This implies that for all but finitely many choices of~$n$, the pair
$\bigl( R_0^n,R_1^n\bigr)$ incorrectly determines the membership
$0^n\in L(A)$ with probability at least $1/(3n)$.
The series
\begin{equation}
  \sum_{n=1}^{\infty} \frac{1}{3n}
\end{equation}
diverges, so by the second Borel--Cantelli lemma the collection $\R$ fails to
compute $L(A)$ for at least one input $0^n$ (and in fact infinitely many
such inputs) with probability~1 for a random choice of
$A\subseteq\Sigma^{\ast}$.
The family $\Q$ is therefore incorrect for a random oracle with probability~1.

Finally, because there are countably many polynomial-time uniform families of
pairs of relativized quantum circuits, and each family correctly decides
$L(A)$ with probability~0, we have that $L(A)\not\in\QSZK^A$ with
probability~1, as required.

%-----------------------------------------------------------------------------%
\subsection*{Acknowledgments}
%-----------------------------------------------------------------------------%

We thank Shalev Ben-David and Robin Kothari for discussions and for sharing
an early draft of their paper with us.
We also thank Rajat Mittal and Abel Molina for helpful conversations.
SM thanks the Department of Computer Science and Engineering at the
Indian Institute of Technology Kanpur, where part of this work was
done, for their hospitality.
This work was partially supported by the Natural Sciences and
Engineering Research Council of Canada, the University of Waterloo
Undergraduate Research Internship program, the President's Scholarship
program, and the Faculty of Mathematics Scholarship program.

%-----------------------------------------------------------------------------%
% Bibliography
%-----------------------------------------------------------------------------%

\bibliographystyle{alpha}
%\bibliography{Oracle}

\begin{thebibliography}{GHMW15}

\bibitem[AH87]{AiellH1987}
William Aiello and Johan H{\aa}stad.
\newblock Perfect zero-knowledge languages can be recognized in two rounds.
\newblock In {\em Proceedings of the 28th Annual {IEEE} Symposium on
  Foundations of Computer Science}, pages 439--448. {IEEE} Computer Society,
  1987.

\bibitem[AKN98]{AharoKN1998}
Dorit Aharonov, Alexei~Y. Kitaev, and Noam Nisan.
\newblock Quantum circuits with mixed states.
\newblock In {\em Proceedings of the 30th Annual {ACM} Symposium on the Theory
  of Computing}, pages 20--30. {ACM}, 1998.
\newblock \arxiv{quant-ph/9806029}.

\bibitem[Amb02]{Ambai2002}
Andris Ambainis.
\newblock Quantum lower bounds by quantum arguments.
\newblock {\em Journal of Computer and System Sciences}, 64(4):750--767, 2002.
\newblock \arxiv{quant-ph/0002066}.

\bibitem[AS04]{AaronS2004}
Scott Aaronson and Yaoyun Shi.
\newblock Quantum lower bounds for the collision and the element distinctness
  problems.
\newblock {\em Journal of the {ACM}}, 51(4):595--605, July 2004.

\bibitem[Bab85]{Babai1985}
L{\'{a}}szl{\'{o}} Babai.
\newblock Trading group theory for randomness.
\newblock In {\em Proceedings of the 17th Annual {ACM} Symposium on Theory of
  Computing}, pages 421--429. {ACM}, 1985.

\bibitem[BBBV97]{BenneBBV1997}
Charles~H. Bennett, Ethan Bernstein, Gilles Brassard, and Umesh~V. Vazirani.
\newblock Strengths and weaknesses of quantum computing.
\newblock {\em {SIAM} Journal on Computing}, 26(5):1510--1523, 1997.
\newblock \arxiv{quant-ph/9701001}.

\bibitem[BDK17]{Ben-daK2017}
Shalev Ben-David and Robin Kothari.
\newblock Quantum sabotage complexity, zero-error algorithms, and statistical
  zero knowledge.
\newblock Manuscript, 2017.

\bibitem[Bei89]{Beige1989}
Richard Beigel.
\newblock On the relativized power of additional accepting paths.
\newblock In {\em Proceedings of the 4th Annual Structure in Complexity Theory
  Conference}, pages 216--224. {IEEE} Computer Society, 1989.

\bibitem[BG81]{BenneG1981}
Charles~H. Bennett and John Gill.
\newblock Relative to a random oracle $\text{A}$, $\text{P}^\text{A} \neq
  \text{NP}^\text{A} \neq \text{co-NP}^\text{A}$ with probability $1$.
\newblock {\em {SIAM} Journal on Computing}, 10(1):96--113, 1981.

\bibitem[BHZ87]{BoppaHZ1987}
Ravi~B. Boppana, Johan H{\aa}stad, and Stathis Zachos.
\newblock {{Does co-NP Have Short Interactive Proofs?}}
\newblock {\em Information Processing Letters}, 25(2):127--132, 1987.

\bibitem[BM88]{BabaiM1988}
L{\'{a}}szl{\'{o}} Babai and Shlomo Moran.
\newblock {{Arthur-Merlin Games}}: {A} randomized proof system, and a hierarchy
  of complexity classes.
\newblock {\em Journal of Computer and System Sciences}, 36(2):254--276, 1988.

\bibitem[Che16]{Chen2016}
Lijie Chen.
\newblock {{A Note on Oracle Separations for BQP}}, 2016.
\newblock Manuscript. \arxiv{1605.00619}.

\bibitem[For89]{Fortn1989}
Lance Fortnow.
\newblock The complexity of perfect zero-knowledge.
\newblock {\em Advances in Computing Research}, 5:327--343, 1989.

\bibitem[FR99]{FortnR1999}
Lance Fortnow and John~D. Rogers.
\newblock Complexity limitations on quantum computation.
\newblock {\em Journal of Computer and System Sciences}, 59(2):240--252, 1999.
\newblock \arxiv{cs.CC/9811023}.

\bibitem[FvdG99]{FuchsG1999}
Christopher~A. Fuchs and Jeroen van~de Graaf.
\newblock Cryptographic distinguishability measures for quantum-mechanical
  states.
\newblock {\em {IEEE} Transactions on Information Theory}, 45(4):1216--1227,
  1999.
\newblock \arxiv{quant-ph/9712042}.

\bibitem[GHMW15]{GutosHMW2015}
Gus Gutoski, Patrick Hayden, Kevin Milner, and Mark~M. Wilde.
\newblock Quantum interactive proofs and the complexity of separability
  testing.
\newblock {\em Theory of Computing}, 11:59--103, 2015.
\newblock \arxiv{1308.5788}.

\bibitem[GMR85]{GoldwMR1985}
Shafi Goldwasser, Silvio Micali, and Charles Rackoff.
\newblock The knowledge complexity of interactive proof-systems (extended
  abstract).
\newblock In {\em Proceedings of the 17th Annual {ACM} Symposium on Theory of
  Computing}, pages 291--304. {ACM}, 1985.

\bibitem[GMR89]{GoldwMR1989}
Shafi Goldwasser, Silvio Micali, and Charles Rackoff.
\newblock The knowledge complexity of interactive proof systems.
\newblock {\em {SIAM} Journal on Computing}, 18(1):186--208, 1989.

\bibitem[HMW13]{HaydeMW2013}
Patrick Hayden, Kevin Milner, and Mark~M. Wilde.
\newblock Two-message quantum interactive proofs and the quantum separability
  problem.
\newblock In {\em Proceedings of the 28th Conference on Computational
  Complexity}, pages 156--167. {IEEE} Computer Society, 2013.
\newblock \arxiv{1211.6120}.

\bibitem[Kob03]{Kobay2003}
Hirotada Kobayashi.
\newblock Non-interactive quantum perfect and statistical zero-knowledge.
\newblock In {\em Proceedings of the 14th International Symposium Algorithms
  and Computation {ISAAC} 2003}, volume 2906 of {\em Lecture Notes in Computer
  Science}, pages 178--188. Springer, 2003.
\newblock \arxiv{quant-ph/0207158}.

\bibitem[LMR{\etalchar{+}}11]{LeeMRSS2011}
Troy Lee, Rajat Mittal, Ben~W. Reichardt, Robert Spalek, and Mario Szegedy.
\newblock Quantum query complexity of state conversion.
\newblock In {\em Proceedings of the 52nd Annual {IEEE} Symposium on
  Foundations of Computer Science}, pages 344--353. {IEEE} Computer Society,
  2011.
\newblock \arxiv{1011.3020}.

\bibitem[{\v S}S06]{SpaleS2006}
Robert {\v S}palek and Mario Szegedy.
\newblock All quantum adversary methods are equivalent.
\newblock {\em Theory of Computing}, 2(1):1--18, 2006.
\newblock \arxiv{quant-ph/0409116}.

\bibitem[Uhl76]{Uhlma1976}
A.~Uhlmann.
\newblock The ``transition probability'' in the state space of a
  {$\ast$}-algebra.
\newblock {\em Reports on Mathematical Physics}, 9:273--279, April 1976.

\bibitem[Val76]{Valia1976}
Leslie~G. Valiant.
\newblock Relative complexity of checking and evaluating.
\newblock {\em Information Processing Letters}, 5(1):20--23, 1976.

\bibitem[Wat02]{Watro2002}
John Watrous.
\newblock Limits on the power of quantum statistical zero-knowledge.
\newblock In {\em Proceedings of the 43rd Annual {IEEE} Symposium on
  Foundations of Computer Science}, page 459. {IEEE} Computer Society, 2002.
\newblock \arxiv{quant-ph/0202111}.

\bibitem[Wat09]{Watro2009}
John Watrous.
\newblock Zero-knowledge against quantum attacks.
\newblock {\em {SIAM} Journal on Computing}, 39(1):25--58, 2009.
\newblock \arxiv{quant-ph/0511020}.

\bibitem[Wat11]{Watro2011}
John Watrous.
\newblock Guest column: an introduction to quantum information and quantum
  circuits.
\newblock {\em {SIGACT} News}, 42(2):52--67, 2011.

\bibitem[Zha05]{Zhang2005}
Shengyu Zhang.
\newblock On the power of {Ambainis} lower bounds.
\newblock {\em Theoretical Computer Science}, 339(2):241 -- 256, 2005.
\newblock \arxiv{quant-ph/0311060}.

\end{thebibliography}

\newcommand{\etalchar}[1]{$^{#1}$}

\end{document}